\documentclass[a4paper,12pt]{amsart}
\usepackage{amssymb,amsmath,mathbbol,mathrsfs}
\usepackage{graphicx}
\usepackage[usenames,dvipsnames]{xcolor}
\usepackage{hyperref}
\usepackage{dsfont} 

\usepackage[left=1.2in,top=1in,right=1.2in,bottom=1in,headheight=0.8in,foot=0.5in]{geometry}
\usepackage[nodisplayskipstretch]{setspace} 
\usepackage{amsaddr}

\usepackage{etoolbox}
\makeatletter
\patchcmd{\@maketitle}
  {\ifx\@empty\@dedicatory}
  {\ifx\@empty\@date \else {\vskip6ex \centering\@date\par\vskip1ex}\fi
   \ifx\@empty\@dedicatory}
  {}{}
\patchcmd{\@adminfootnotes}
  {\ifx\@empty\@date\else \@footnotetext{\@setdate}\fi}
  {}{}{}
\makeatother

\usepackage{natbib}
\setcitestyle{aysep={}} 

\usepackage{marginnote}

\makeatletter
\let\uppercasenonmath\@gobble

\hypersetup{
	colorlinks=true,         
	linkcolor=MidnightBlue,          
	citecolor=MidnightBlue,
	urlcolor=MidnightBlue            
 }
      
\newcommand{\be}{\begin{equation}} 
\newcommand{\ee}{\end{equation}}
\newcommand{\D}{{\mathrm{D}}}
\newcommand{\G}{{\mathcal{G}}}
\newcommand{\pp}{{\partial}}
\newcommand{\fG}{{\mathrm{Lie}(\G)}}

\newcommand{\RR}{\mathds{R}} 
\newcommand{\HH}{\mathcal{H}}

\usepackage{thmtools} 
\usepackage{thm-restate}

\title{The Gauge Argument: A Noether Reason}
\date{\today\\Forthcoming in \emph{The Physics and Philosophy of Noether's Theorems}, Read, Roberts and Teh (Editors), Cambridge University Press}

\author{Henrique Gomes}
\address{\vspace{-0.8pc}University of Cambridge\\Trinity College, CB2 1TQ, United Kingdom}
\email{\href{mailto:gomes.ha@gmail.com}{gomes.ha@gmail.com}}

\author{Bryan W. Roberts} 
\address{\vspace{-0.8pc}London School of Economics \& Political Science\\Houghton Street, London WC2A 2AE, United Kingdom}
\email{\href{mailto:b.w.roberts@lse.ac.uk}{b.w.roberts@lse.ac.uk}}

\author{Jeremy Butterfield}
\address{\vspace{-0.8pc}University of Cambridge\\Trinity College, CB2 1TQ, United Kingdom}
\email{\href{mailto:jb56@cam.ac.uk}{jb56@cam.ac.uk}}

\begin{document}
\setstretch{1.0}
\maketitle 

\begin{abstract}
Why is gauge symmetry so important in modern physics, given that one must eliminate it when interpreting what the theory represents? In this paper we discuss the sense in which gauge symmetry can be fruitfully applied to constrain the space of possible dynamical models in such a way that forces and charges are appropriately coupled. We review the most well-known application of this kind, known as the `gauge argument' or `gauge principle', discuss its difficulties, and then reconstruct the gauge argument as a valid theorem in quantum theory. We then present what we take to be a better and more general gauge argument, based on Noether's second theorem in classical Lagrangian field theory, and argue that this provides a more appropriate framework for understanding how gauge symmetry helps to constrain the dynamics of physical theories.
\end{abstract}


\setstretch{1.4}
\section{Introduction}\label{sec:intro}

All interpretations of modern gauge theories adopt two core assumptions at their foundation. The first is that gauge symmetry arises when there are more variables in a theory than there are physical degrees of freedom. Hence the well-known soubriquets: gauge is `descriptive redundancy', `surplus structure', and `descriptive fluff'. Correspondingly, considerable effort has been devoted to techniques for eliminating gauge redundancy in order to appropriately interpret gauge theories.\footnote{Cf.\ \citet{earman2002g,earman2003t,earman2004l}, \citet{healey2007book} and \citet{RosenstockWeatherall2016c,rosenstockweatherall2018e}. To some extent we agree: see \citet{gomesriello2020theta} in response to \citet{dougherty2020cp}.} The second assumption is that a theory with gauge symmetry constitutes the gold standard of a modern physical theory: witness the gauge symmetries invoked in the Standard Model. This leads to a remarkable \emph{puzzle of gauge symmetry:} why is gauge symmetry so ubiquitous? We do not aspire to give a single, ultimate answer. The purpose of this paper is to articulate one answer to this question: namely, that gauge symmetry provides a path to building appropriate dynamical theories --- and that this rationale invokes the theorems of Emmy \citet{noether1918a}.

Of course, a number of answers --- alternatives to simple eliminativist interpretations of gauge --- have already been articulated. For example, gauge theories provide a convenient calculational technique, as in the use of a potential $A$ in classical electromagnetism. On the other hand, it would be odd if the ``great gauge revolution'' turned out to be, {\em au fond}, a matter of calculational convenience; and indeed there is more to the story. More importantly: in many cases gauge symmetry cannot be eliminated without also eliminating the possibility of local Lorentz invariance and Lorentz invariant quantities.

Gauge symmetries can also encode important physical information, in spite of their being symmetries. The best-known example, which is vivid because of its experimental significance, is the effect of \citet{aharonovbohm1959a}, in which non-local information (of a different kind than due to quantum entanglement) can be contained in the gauge potential.\footnote{Other more recent examples include \citet{rovelli2014wg,rovelli2020g} and \citet{gomes2019b,gomes2020des,gomes2020h}, who emphasise the role of gauge in characterising the coupling of systems and regions, respectively; and \citet{nguyentehwells2020g}, who emphasise its role in defining certain local gauge fields. For a history of the early debate on the AB Effect see \citet{hiley2013ab}. Philosophers have recently focused on questions about the locality and reality of the gauge potential \cite[cf.][]{healey1997ab,belot1998em,maudlin1998ab,healey1999ab,nounou2003ab,mattingly2006gm,healey2007book,lyre2009ab,pitt-2009discussion,myrvold2011s,wallace2014ab} and \citet{mulder2021g}. More recently, \citet{shech2018ab} and \citet{earman2019ab} have challenged the idealisations associated with the Aharonov-Bohm effect, and \citet{dougherty2020n} has defended them.} But we will be concerned with the two theorems of Emmy \citet{noether1918a}.\footnote{For details on the historical development of Noether's theorems see \citet{ks2011noether}. For a modern statement of the first and second theorems, cf.\ \citet{olver1993}, Theorems 5.58 (p. 334) and 5.66 (p. 343) respectively.} 
 
Noether's first and better-known theorem (commonly called simply \emph{Noether's theorem}) implies that global (or what we will call {\em rigid}) symmetries of a classical Lagrangian field theory --- i.e. symmetries in which the redundancy is specified in exactly the same way at all spacetime points --- correspond to charges that are conserved over time, such as energy and angular momentum. For example, the conservation of an electron's charge can be viewed as arising from the (redundant) global phases of the electron's wavefunction. But we will be equally concerned with Noether's second theorem, which is about  local (or what we will call {\em malleable}) gauge symmetries --- meaning that the specified redundancy varies between spacetime points. Agreed: this theorem's physical significance is of course already well recognized, including in the philosophical literature (\citet{bradingbrown2000n,bradingbrown2003ch}, \citet{brading2002s}). In particular, a recent line of work shows how such malleable gauge symmetries encode relationships between spatial or spacetime regions, and thus between parts and wholes in a field theory.\footnote{See \citet{friedeldonnelly2016l}, \citet{gomes2020des,gomes2020h} and \citet{gomesriello2020}, in response to the discussions of `direct empirical significance' in \citet{bradingbrown2004g} and \citet{greaveswallace2014e}.}

In this paper, we will urge that these two theorems give us a further answer to the puzzle, `why gauge?' It is an established, indeed conventional, answer amongst practising physicists. For it is implicit in the well-known \emph{gauge argument} or the \emph{gauge principle} first formulated by Hermann \citet{weyl1929eug}.\footnote{An English translation and extended commentary is given by \citet[Chapter 5]{oraifeartaigh1997}.} This argument begins with an assumption of local gauge symmetry, and then claims to `derive' the form of the dynamics of quantum theory in a way that exhibits `minimal coupling' to an electromagnetic potential. We claim that this is an instance of a much more general role for gauge, which has not been at all discussed in the philosophical literature: gauge symmetry supports theory construction, in particular by constraining the space of models to those in which charges appropriately couple to forces. Although some philosophers like \citet{bradingbrown2003ch} have pointed out the role of gauge symmetry in theory construction, it is this last coupling of charges to forces that we would like to highlight, which provides the answer to the puzzle of gauge symmetry that we will advocate here.

As experts will be quick to note: the gauge argument in its common textbook form is fraught with difficulties. However, our argument is that these difficulties can be overcome; and indeed that there is a more general gauge argument available for use in the construction of physical theories. We thus proceed in Section \ref{sec:ga-and-woes} to rehearse the usual gauge argument and its woes. In Section \ref{sec:ga-as-a-thm} we offer a glimmer of hope, by reconstructing the gauge argument as a theorem of quantum theory, which we argue vindicates to some extent its use in the formulation of quantum electrodynamics.

The real limitation of the textbook gauge argument, as we shall see, is that it does not reflect the generality of the kind of argument that physicists typically use. Thus, in Section \ref{sec:ANoether}, we present a much more general gauge argument, which we call the \emph{Noether gauge argument}, in the context of classical Lagrangian field theory. The key to understanding this argument is the combined use of \emph{both} Noether's first \emph{and} second theorem. In the first step, one applies Noether's first theorem to establish the conservation of charge. In the second step, one makes use of the power of Noether's second theorem, to infer specific interpretive information about how these charges couple to gauge fields. We draw out and clarify what that information is, in the presence of various kinds of symmetries that are sometimes referred to as `gauge', in order to illustrate the precise extent to which the gauge argument can be fruitfully used to constrain physical theories. 

\section{The gauge argument and its critics}\label{sec:ga-and-woes}

The textbook gauge argument or gauge principle uses gauge invariance to motivate a quantum theory of electromagnetism. We begin Section \ref{subs:beware} with a brief presentation of this argument as it is usually presented. Classic textbook statements can be found in \citet[\S 6.14]{schutz1980g} \citet[\S 4.2]{gc1989dg}, and \citet[\S 3.3]{ryder1996qft}, among many other places. Then in Section \ref{subs:criticisms} we assess it. The argument has been discussed in the form below by philosophers as well, such as \citet{teller1997m,teller2000gauge}, \citet{brown1998a}, \citet{martin2002g}, and \citet[\S 2]{wallace2009anti}. In spite of the criticisms, we will argue in Section \ref{sec:ga-as-a-thm} that a grain of truth remains in the gauge argument.

\subsection{Beware: Dubious arguments ahead}\label{subs:beware}

We begin by describing a quantum system with the Hilbert space $L^2(\RR^3)$ of wavefunctions, recalling that a unique pure quantum state is represented not by vector, but by a `ray' of vectors related by a complex unit. This implies that the transformation $\psi(x)\mapsto e^{i\theta}\psi(x)$ for some $\theta\in\RR$, referred to as a `global phase' transformation, acts identically on rays, and is in this sense an invariance of the quantum system. But now, the story goes, suppose we replace this with a `local phase' transformation $\psi(x)\mapsto e^{i\phi(x)}\psi(x)$, in which the constant $\theta$ is replaced with a function $\phi:\RR^3\rightarrow\RR$, or indeed with a smooth one-parameter family of such functions $\phi_t(x)$ for each $t\in\RR$. This transformation is `local' in the sense that its values vary smoothly across space and time. The corresponding Hilbert space map $W_\phi:\psi\mapsto e^{i\phi}\psi$ does not act identically on rays. However, one might still wish to postulate that this transformation has no `physical effect' on the system, or is `gauge'. Various motivations for this step are given in the textbooks, often with vague references to general covariance of the kind found in general relativity: which we will return to shortly. But to mimic the standard presentation, we will simply press forward, referring to $W_\phi:\psi\mapsto e^{i\phi}\psi$ as a \emph{local} or \emph{malleable} \emph{gauge transformation}.

The main premise of the argument is to assume that the Schr\"odinger equation must be invariant under this local phase transformation. But, for the free non-relativistic Hamiltonian in the Schr\"odinger (position) representation, this is not the case.\footnote{Obvious variations of the argument exist for relativistic wave equations too \citep[cf.][\S 3.3]{ryder1996qft}.} Writing $\psi_t(x):=e^{-itH}\psi(x)$ with $H = \tfrac{1}{2m}P^2$, one finds that $W_\phi:\psi\mapsto e^{i\phi_t(x)}\psi$ transforms the Schr\"odinger equation to $i\tfrac{d}{dt}\left(e^{i\phi_t(x)}\psi_t(x)\right) = \tfrac{1}{2m}P^2e^{i\phi_t(x)}\psi_t(x)$, which is equivalent\footnote{The LHS is $i\tfrac{d}{dt}e^{i\phi_t(x)}\psi_t(x) = e^{i\phi_t(x)}\left(-\tfrac{d\phi}{dt} + i\tfrac{d}{dt}\right)\psi_t(x)$. For the RHS, use the fact that $e^{-i\phi_t(x)}Pe^{i\phi_t(x)}=P+\nabla\phi_t(x)$ (cf.\ Footnote \ref{fn-bwr-p-calc}), and so $e^{-i\phi_t(x)}P^2e^{i\phi_t(x)} = (e^{-i\phi_t(x)}Pe^{i\phi_t(x)})^2 = (P+\nabla\phi_t(x))^2$. Thus the RHS is $\tfrac{1}{2m}P^2e^{i\phi_t(x)}\psi_t(x) = e^{i\phi_t(x)}\tfrac{1}{2m}(P+\nabla\phi_t(x))^2\psi_t(x)$. Multiplying both sides on the left by $e^{-i\phi_t(x)}$ and rearranging then gives the result.\label{eq1-calc}} to the statement that,
\begin{equation}\label{g-tr-schr}
 i\tfrac{d}{dt}\psi_t(x) = \left(\tfrac{1}{2m}(P+\nabla\phi_t)^2 + \tfrac{d\phi_t}{dt}\right)\psi_t(x).
\end{equation}
Instead of preserving the Schr\"odinger equation, a gauge transformation produces the additional terms $\nabla\phi_t$ and $\tfrac{d\phi_t}{dt}$ in the Hamiltonian.

To correct this situation, the big move of the gauge argument is to introduce a vector  $A = (A_1,A_2,A_3)$ and a scalar $V$, which are assumed to behave under the gauge transformation as,
\begin{align}\label{AV-trans-rules}
  A\mapsto A+\nabla\phi_t, && V\mapsto V-\tfrac{d\phi_t}{dt}.
\end{align}
This has the form of the familiar gauge freedom of the electromagnetic four-potential that leaves the electromagnetic field unchanged.

To restore invariance of the Schr\"odinger equation under gauge transformations, one thus apparently needs only to assume that the Hamiltonian is not free, but rather given by,
\begin{equation}\label{eq:mc-ham}
  H = \sum_{r=1}^3\tfrac{1}{2m}(P_r-A_r)^2 + V,
\end{equation}
which is known as the \emph{minimally coupled} Hamiltonian. For, replacing the Hamiltonian in the Schr\"odinger equation with this one, we find that the transformation rules for $A$ and $V$ perfectly compensate for the extra terms appearing in Equation \eqref{g-tr-schr}. Thus, gauge invariance of the Schr\"odinger equation is obtained, provided the Hamiltonian contains interaction terms $A$ and $V$ that behave like the 3-vector potential $A$ and scalar potential $V$ for an electromagnetic field.

With an eye towards a modern gauge theory formulated as a vector bundle with a derivative operator, it is even possible to interpret the potentials $A$ and $V$ as associated with a change of derivative operator: writing $\pp_\mu := (\tfrac{d}{dt},\nabla)$ and $A_\mu = (V,A)$, one finds that the procedure above is equivalent to replacing $\pp_\mu$ with,
\begin{equation}\label{eq:cov-der}
  D_\mu := \pp_\mu + iA_\mu = (\tfrac{d}{dt} + iV,\nabla + iA) = (D_t,D).
\end{equation}
This is commonly referred to as a `covariant derivative'. Then, substituting $\tfrac{d}{dt}\mapsto D_t$ and $\nabla\mapsto D$ into the free Schr\"odinger Equation $i\tfrac{d}{dt}\psi = \tfrac{1}{2m}\nabla^2\psi$ and rearranging, we derive the minimally coupled Hamiltonian of Equation \eqref{eq:mc-ham}. Accordingly, this choice of Hamiltonian is sometimes advocated, for example by \citet{lyre2000gep}, on the basis of a `generalised equivalence principle', according to which electromagnetic interactions with all matter fields ``can be transformed away''.\footnote{This principle arises in particular on a principal fibre bundle formulation of gauge theory; for philosophical appraisals, see \citet{lyre2000gep}, \citet[\S 5]{weatherall2016ymgr}, and \citet[Ch. 6.3]{healey2007book}.} In short, it appears as if minimal electromagnetic coupling has been derived out of nothing: or at least, from an assumption of gauge invariance.
  
\subsection{Criticisms of the gauge argument}\label{subs:criticisms}

That is how the story is usually presented. We agree: it is far from water-tight. The argument begins with a system with a global symmetry, gratuitously generalises it to a local symmetry --- which, to emphasise, was not required for mathematical consistency or for empirical adequacy --- and then, in order to fix the ensuing non-invariance of the governing equations, proceeds to conjecture a new force of nature. To put it uncharitably: the argument fixes a problem that didn't exist by conjecturing a redundant field, and then turns this game around, claiming to come out successfully by `retrodicting' the existence of electromagnetism. More charitably: the gauge argument suffers from at least three categories of concerns. We will set out each of these three concerns here. In Section \ref{sec:ga-as-a-thm} we will then offer a glimpse of how the first two can be answered, and in Section \ref{sec:ANoether} present an alternative {\em Noether gauge argument} that answers them entirely.

The first category of concerns is the gauge argument's claim to have derived a dynamics that is specifically electromagnetic in nature. Although a formal set of operators $A_\mu=(V,A)$ have been included in the dynamics, no evidence is given that these operators take the form required for any \emph{specific} electromagnetic potential, or that the coupling to $A_\mu$ will be proportional to a particle's charge $e$, or even that $A_\mu$ is non-zero. And if they could be shown to be non-zero, then as \citet[p.210]{wallace2009anti} rightly asks: ``how do neutral particles fit into the argument?'' A minimally coupled dynamics does not to apply to neutral particles, and yet since the gauge argument never mentioned or assumed anything about charge, it presumably is intended to apply to them.

This concern can be assuaged by scaling back the conclusion of the gauge argument: its aim is not to derive any particular electromagnetic interaction, but rather to \emph{constrain} the dynamics so as to be compatible with gauge invariance. This leaves open the specific character of $A_\mu$, and indeed even the question of whether it is zero. Although not all authors adopt this attitude towards the gauge argument, we advocate it as the preferable attitude, and will develop it in more detail in the subsequent Sections.

A second category of problems arises out of the free-wheeling argumentative style of the gauge argument. For example, it is not a strict deductive derivation of either the electromagnetic potential or the dynamics. At best, the gauge argument appears to show that one \emph{can} adopt a minimally coupled Hamiltonian in order to assure gauge invariance. But this does not ensure that one \emph{must} do so: the door appears to be left open for other dynamics to be gauge invariant, but without taking the minimally coupled form that the gauge argument advocates. As \citet[p.S230]{martin2002g} writes: ``The most I think we can safely say is that the form of the dynamics characteristic of successful physical (gauge) theories is \emph{suggested} through running the gauge argument.''

Another example of free-wheeling argumentation is in the motivation for requiring the local gauge transformations $W_\phi:\psi\mapsto e^{i\phi} \psi$ to be symmetries. Sometimes a preference for this transformation over global phase transformations is dubiously motivated by a desire to avoid superluminal signalling.\footnote{For example, \citet[p.93]{ryder1996qft} writes: ``when we perform a rotation in the internal space of $\phi$ at one point, through an angle $\Lambda$, we must perform the same rotation at all other points at the same time. If we take this physical interpretation seriously, we see that it is impossible to fulfil, since it contradicts the letter and spirit of relativity, according to which there must be a minimum time delay equal to the time of light travel.'' For a detailed critique, see \citet[p.S227]{martin2002g}.} In other cases it is motivated by the coordinate invariance of a spatial coordinate system. But as \citet[p.210]{wallace2009anti} points out, no reason is given as to why we do not similarly consider local transformations of configuration space, momentum space, or any other space, to be symmetries. Nor is there any clear reason why the $U(1)$ symmetry of electromagnetism is chosen as the global symmetry motivating the move to the local symmetry, as opposed (say) the $SU(3)$ symmetry of the strong nuclear force.

We will claim that most of these problems can be entirely solved. In the first place, the gauge argument can in fact be tightened and turned into a valid derivation, as we will show in the next Section. Not only is it that one \emph{can} adopt the minimally coupled Hamiltonian in the presence of gauge invariance, but one \emph{must} do so, when gauge invariance is viewed in terms of a particular constraint on the `velocity observable', in a sense we will precisely define. We will similarly argue that the postulate that local gauge transformations are symmetries in a certain sense --- namely, that they are unitary operators --- is not really a postulate, but a formal fact about the framework in which these transformations are presented.

Regarding the generalisation of the gauge argument to other global symmetry groups beyond electromagnetism, we wholeheartedly agree with Wallace: one should expect, and indeed we will argue in Section \ref{sec:ANoether}, that an appropriate generalisation of the gauge argument can also be applied to these more general gauge groups.

Our approach here speaks to a third category of concerns, that the gauge argument is awkwardly placed as an argument for a quantum theory of electromagnetism. The construction of a covariant derivative operator suggested by the gauge argument is most appropriately carried out not in quantum field theory, but in the classical Yang-Mills theory of principal fibre bundles. Here too we agree with Wallace:
\begin{quote}
  ``In fact, it seems to me that the standard argument feels convincing only because, when using it, we forget what the wavefunction really is. It is not a complex classical field on spacetime, yet the standard argument, in effect, assumes that it is. This in turn suggests that the true home of the gauge argument is not non-relativistic quantum mechanics, but classical field theory.'' \citep[p.211]{wallace2009anti}
\end{quote}
Indeed, it is remarkable that in the presentation of the gauge argument above, the role of the `rigid' or `global' $U(1)$ symmetry is hardly substantial: only the local malleable symmetries play any substantial role in the argument. This is an oddity to be sure, though one that we will correct shortly.

In Section \ref{sec:ANoether}, we will switch perspectives from the \emph{verdammten Quantenspringerei} to the context of classical Lagrangian field theory, and propose a framework that substantially clarifies the roles of rigid gauge symmetries, of malleable gauge symmetries, and of their relationship, which we will call the `Noether gauge argument'.

But before we develop this argument, we would like to first show how the first two categories of concern are unwarranted: a more rigorous formulation of the textbook gauge argument is possible, which dispels any worries about a free-wheeling, under-motivated argument. Its only serious shortcoming in this more rigorous form, as we will see, is its lack of generality. This will be corrected in Section \ref{sec:ANoether}, when we develop an analogous argument in classical field theory. It is not a straightforward generalisation of Section \ref{sec:ga-as-a-thm}'s theorem, but rather a `cousin' of it. Namely, the theme of both results is that malleable gauge symmetries are used to constrain the dynamics. 

\section{The gauge argument as a theorem}\label{sec:ga-as-a-thm} 

\subsection{Gauge transformations and probabilities}

Suppose we view the textbook gauge argument as beginning with a strongly continuous one-parameter unitary representation $t\mapsto U(t)=e^{-itH}$ on the Hilbert space $\HH=L^2(\RR^3)$, but we do not yet know the form of the Hamiltonian $H$. Here we will adopt the vector notation $x=(x_1,x_2,x_3)$, and similarly $Q=(Q_1,Q_2,Q_3)$, with the latter representing the vector position, where for each $i=1,2,3$ we have a self-adjoint operator $Q_i$ defined by $Q_i\psi(x)=x_i\psi(x)$. The aim of the gauge argument, as we see it, is then to use gauge symmetry to constrain the possible dynamics that are available.

The first step is to observe a fact in this framework --- \emph{not} an assumption, but a consequence! Namely: the local gauge transformations are symmetry operators that transform position and momentum as one would expect a gauge transformation to do.\footnote{By `symmetry' we mean \emph{unitary} operators, sometimes referred to as `kinematic symmetries' \citep[cf.][Chapter 13]{jauchQM}. Antiunitary operators are symmetries of time-reversing transformations \citep[see][]{roberts2017a,roberts2021tr}, but will play no role in this discussion.} That is: for any smooth function $\phi:\RR^3\rightarrow\RR$, whose gradient $\nabla\phi$ we can think of as smoothly deforming space however we wish, there is a Hilbert space operator that preserves inner products, i.e. a unitary operator, given by:
\begin{equation}\label{eq:kin-mal-3d}
  W_{\phi}\psi(x) := e^{i\phi(x)}\psi(x)
\end{equation}
for all $x\in\RR^3$. A similar fact applies in four-dimensions: defining a one-parameter set $\phi_t(x)$ of such transformations, there is a `malleable symmetry' or  unitary $W_{\phi}$ defined by,
\begin{equation}\label{eq:kin-mal-4d}
  W_{\phi_t}\psi(x) := e^{i\phi_t(x)}\psi(x).
\end{equation}
for all $\psi\in\HH$, and for all $x\in\RR^3$ and $t\in\RR$. These unitaries $W_{\phi}$ transform position $Q=x$ and momentum $P =-i\nabla$ as one would expect a gauge transformation to do:\footnote{The first equation follows just from $\phi$ being a function of position. The second equation follows from the fact that $PW_{\phi}^*\psi = -i\nabla \left(e^{-i\phi}\psi\right) = -i\left(e^{-i\phi}\nabla\psi - i(\nabla \phi)e^{-i\phi}\psi\right) = e^{-i\phi}(-i\nabla - \nabla \phi)\psi = W_\phi^*(P - \nabla\phi)\psi$, for all $\psi\in L^2(\RR^3)$.\label{fn-bwr-p-calc}}
\begin{align}\label{eq:gauge-momentum-transf}
  W_\phi Q W^*_\phi = Q &&  W_\phi P W_\phi^* = P - \nabla \phi.
\end{align}

Some presentations (rather mysteriously) even say that these facts require a postulate of general relativity; as when \citet[p.219]{schutz1980g} writes, ``[t]hese turn out to be the appropriate equations for that system in \emph{general} relativity'' (cf.\ also \citet[p.11]{lyre2000gep}). 

We would like to emphasise that, for the purpose of constructing a unitary operator that implements the transformations in Equation \eqref{eq:gauge-momentum-transf}, \emph{no such arguments are needed}. It is a direct mathematical consequence of the Hilbert space formalism that each malleable transformation $\phi$ gives rise to a unitary $W_\phi=e^{i\phi}$, and that these unitaries behave like gauge transformations with respect to position and momentum. A new physical postulate is only required when malleable transformations are used to constrain the dynamics: which we consider next.

\subsection{A rigorous quantum gauge argument}

We are now ready to view the gauge argument as providing a constraint on the possible dynamics for quantum theory, on the basis of a certain assumption of gauge invariance. The postulate of gauge invariance that we will state is in the spirit of the textbook gauge argument. We begin by clarifying exactly what we take this to mean.

As before, we begin with a unitary dynamics $t\mapsto U_t = e^{-itH}$ on the Hilbert space $\HH=L^2(\RR^3)$, for which we do not yet know the Hamiltonian $H$. The 3-velocity in this context is defined to be the rate of change of the position operator in the Heisenberg picture with respect to this unitary group. Thus, it is given by,
\begin{equation}\label{eq:velocity-obs}
  \dot{Q} := \tfrac{d}{dt}U^*_tQU_t\big|_{t=0} = i[H,Q],
\end{equation}
where $Q=(Q_1,Q_2,Q_3)$ is defined as in the Schr\"odinger representation by $Q_r\psi(x)=x_r\psi(x)$ for each $r=1,2,3$, for all wavefunctions $\psi$ in its domain.

Our central postulate for how malleable symmetries constrain the dynamics, in pictorial terms, is that if a smooth transformation $\phi:\RR^3\rightarrow\RR$ is applied, then the 3-velocity should compensate for the extra term $\nabla\phi$ that arises, as shown in Figure \ref{fig:malleable}. One must remove the contribution of $\nabla\phi$ from the value of the velocity.\footnote{In Equation \eqref{eq:gauge_trans} of Section \ref{sec:ANoether} we propose a similar postulate to characterise gauge invariance in the context of classical field theory; where we note that it is also possible to consider the contribution of higher-order terms, in which case this definition can be viewed as a first-order approximation.} Thus, if $W_\phi=e^{i\phi}$ is the unitary transformation associated with $\phi$, then the assumption says:
\begin{equation}\label{eq:3vel-ga}
  e^{i\phi}\dot{Q}e^{-i\phi} = \dot{Q} - \nabla\phi.
\end{equation}
This is close enough to the meaning of a transformation by a smooth function that it is somewhat surprising that it provides any constraint on the dynamics, let alone enough to establish minimal coupling. 

\begin{figure}[tbh]\begin{center}
    \includegraphics[width=0.7\textwidth]{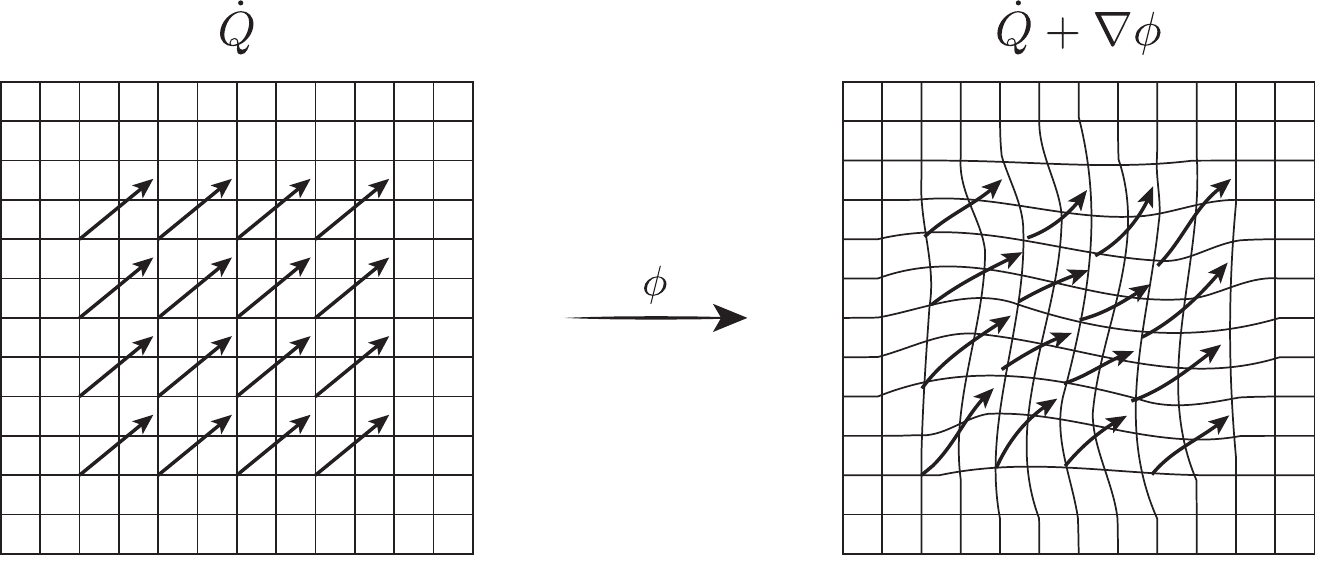}
    \caption{The constraint of malleable symmetries is that the 3-velocity $\dot{Q}$ must transform so as to remove the contribution of the gradient $\nabla\phi$ of a smooth function.}\label{fig:malleable}
  \end{center}\end{figure}

Although Equation \eqref{eq:3vel-ga} is thus not a priori, it is natural and rich in consequences. We formalise this in the following.\footnote{A proof is provided in the \hyperref[sec:appendix]{Appendix}. The ideas of this proof are a turned-around cousin of \citet{jauch1964a} and \citet[\S 13-5]{jauchQM}, an analysis of which is in \citet[Appendix A]{roberts-dissertation}.}

\begin{restatable}{thm}{GAhamiltonian}\label{thm:ga} 
In the Schr\"odinger (position) representation on $L^2(\RR^3)$ with dynamics $t\mapsto U_t=e^{-itH}$, let $\dot{Q}_r := \tfrac{d}{dt}(U_t^*QU_t)\big|_{t=0} = i[H,Q_r]$. Suppose that for every smooth one-parameter set of smooth functions $\phi_t:\RR^3\rightarrow\RR$, at time $t=0$,
  \begin{equation}\label{WAssptnThm1}
    W_\phi\dot{Q}W_\phi^* = \dot{Q}-\nabla\phi,
  \end{equation}
where $W_\phi := e^{i\phi_t}$ and $\nabla = (\pp_{x}, \pp_{y}, \pp_{z})$. Then:

  (i) the Hamiltonian is constrained to be of the form,
\begin{equation}\label{eq:ham-form}
  H = \sum_{r=1}^3\tfrac{1}{2m}(P_r - A_r)^2 + V
\end{equation}
for some $m>0$ and some functions (i.e. operators) $A_r(x)$ and $V(x)$ that depend only on position, and not on momentum;

(ii) the transformed unitary group $t\mapsto \tilde{U}_t := W_\phi U_t =: e^{-itH_\phi}$ (where the last equation defines $H_\phi$ as the generator of $\tilde{U}_t$, by Stone's theorem) is given by,
  \begin{equation}\label{eq:gauge}
      H_\phi = \sum_{r=1}^3\tfrac{1}{2m}\left(P_r - A_r^\phi\right)^2 + V^\phi,
    \end{equation}
    where $A_r^\phi = A_r + \nabla\phi$ and $V^\phi = V -\tfrac{d\phi}{dt}$; and

    (iii) the Schr\"odinger equation is `gauge invariant' under $W_\phi$, in that if $i\tfrac{d}{dt}\psi(t) = H\psi(t)$, then $\psi_\phi(t) := W_\phi\psi(t)$ satisfies $i\tfrac{d}{dt}\psi_\phi(t) = H_\phi\psi_\phi(t)$. 
\end{restatable}

In summary: we assume that the 3-velocity transforms  in a given manner under smooth transformations of space ($W_\phi\dot{Q}W_\phi^* = \dot{Q}-\nabla\phi$), and prove a consequence that the dynamics \textit{are constrained} to be given by a minimally coupled Hamiltonian: in which, moreover,  the smooth transformations behave like gauge transformations of the vector and scalar potentials, in that they leave the Schr\"odinger equation invariant.

We take this to largely dispel the concerns about the free-wheeling nature of the textbook gauge argument. Cast in this form, any Hamiltonian that is constrained by malleable transformations as in Equation \eqref{WAssptnThm1} \emph{must} be minimally coupled; and this leads to gauge invariance as in (ii) and (iii) of the Theorem. Of course, this still leaves the specific coupling undetermined up to a choice of vector and scalar potential, including the choice of zero. But this is fully compatible with our perspective that gauge symmetry can be viewed as providing a powerful constraint on the dynamics.

There remains the question of how to view this kind of argument in the context of more general gauge groups. And there remains a question of what role the rigid gauge transformations play in constraining a theory's dynamics. Here we propose a shift to a more natural perspective for dealing with such constraints, which uses the tools of classical Lagrangian field theory. This is the subject of the remainder of this paper.

\section{A Noether Reason for Gauge}\label{sec:ANoether}

\subsection{Overview}

For a more general view of how gauge symmetries constrain the dynamics of a physical theory, we will now, as announced in Section \ref{sec:intro},  make a two-step use of the theorems of Emmy \citet{noether1918a}: the first, and then the second. We will refer to this as the \emph{Noether gauge argument}. Agreed: this is by no means a new observation, since practising physicists use this property of gauge frequently!\footnote{A succinct example is \citet{averyschwab2016n}, who write ``Noether’s second theorem, which constrains the general structure of theories with local symmetry''.} But we believe it is worth highlighting and clarifying exactly the kind of information that can be extracted in various cases, as part of our advocacy (cf.\ Section \ref{sec:intro}) that philosophical discussions of gauge should better recognise gauge's significance for theory construction.

The Noether gauge argument proceeds in two steps. First, we choose a rigid gauge symmetry associated with an arbitrary global gauge group, and propose that its action produces a variational symmetry: by Noether's first theorem, this guarantees the presence of a collection of conserved quantities. But  matter fields do not exist in  isolation: they couple to other `force' fields, and possibly to long-range ones. Thus, in the second step, we introduce such a field and apply Noether's second theorem, `loosening' the rigid symmetries to malleable ones; and we show that this provides three  concrete constraints on the dynamics, viz.\ the vanishing of the three lines in Equation \eqref{eq:vanishing} below. This result is limited: we make no claim that the only way to get these benefits is by malleable symmetries. However, it remains the best known and most tractable way to achieve them.

The interpretation of these constraints can be seen on a case-by-case, or sector-by-sector, basis: we consider their implications for rigid versus malleable symmetries, as well as for $A$-independent versus $A$-dependent ones. Thus in the following Sections we will spell out the consequences of the three constraints for four different sectors of the theory.  In particular, we will find through explicit computation --- adopting only a minor additional assumption of non-derivative coupling --- that when we couple the matter fields to force fields, gauge-invariance guarantees that the Lagrangian for these fields is massless, and so they constitute long-range interactions.\footnote{The formalism equally applies to spin-2, or gravitational, fields; but, apart from some cursory remarks, we will not discuss these.} The generalised Gauss laws thus are guaranteed to relate the content of the matter current within a region to the flux of the other force fields at distant closed surfaces surrounding such a region.

Disclaimers: first, in the interest of clarity and pedagogy, we will not try to incorporate the full generality of Noether's theorems, which is truly extraordinary but over-complicated for our discussion. In its place we make several simplifying assumptions, both about the Lagrangian density and about the action of the gauge group, which are not strictly speaking necessary but which simplify our argument. Second, throughout this discussion, we will follow standard practice and distinguish two equivalence relations for classical fields on a manifold. First, we write `$=$' to denote ordinary equality between fields, irrespective of the satisfaction of the equations of motion, and refer to this as \emph{strong} or \emph{off-shell} equality. Second, given a fixed Lagrangian, we write `$\approx$' to denote equality between fields that holds if the Euler-Lagrange equations are satisfied for that Lagrangian, and refer to this as \emph{weak} or \emph{on-shell} equality.\footnote{This common terminology is due to Dirac \citep[cf.][]{HenneauxTeitelboim1994}.}

\subsection{Field theory and malleable symmetries}

In the textbook gauge argument, the (rigid) gauge transformations were assumed to be the group of phase transformations $U(1)$, and the (malleable) gauge group was given by its locally-varying analogues. We now lift this restriction and allow the global gauge group to be any compact Lie group $G$, with  Lie algebra $\mathfrak{g}$. We can characterize the malleable (infinitesimal) gauge transformations induced by $G$ (and $\mathfrak{g}$), on a local patch $U$, as the smooth functions from $U$ to $G$ (to $\mathfrak{g}$, respectively).\footnote{The  precise description of gauge transformations requires a brief incursion into the mathematical theory of fiber bundles. A principal fiber bundle is a smooth manifold $P$ that admits a smooth and free action of the group $G$, e.g. $G\times P\rightarrow P\,;\,\,  (g, p)\mapsto g\cdot p$. Spacetime is encoded in the bundle as the collection of orbits $\mathcal{O}$ of the group: $\pi:P\rightarrow M\simeq P/G$, where $\pi$ is the smooth submersion induced by the action of $G$. One can think of the bundle as attaching an orbit space isomorphic to $G$ (but without a preferred identity, much like an affine space) to each point of $M$. Thus $\pi^{-1}(x)=:\mathcal{O}_p=\{g\cdot p,\,\, g\in G\}\simeq G$ for $x\in M$,  but there is no canonical isomorphism between $\pi^{-1}(x)$ and $G$ (it requires the choice of a point $p$) and no  decomposition $P\simeq M\times G$. To talk about gauge transformations as maps in spacetime, we must restrict their domains to subsets of spacetime. Locally, i.e. for subsets $U\subset M$, we can write $\pi^{-1}(U)\simeq U\times G$, once we have choice of trivialisation: $s:U\rightarrow P$ such that $\pi\circ s=\mathrm{Id}_U$. Here $s$ can be thought of as a submanifold that intersects all the group orbits in $\pi^{-1}(U)$ only once. In terms of $s$ we can write the trivialization as $U\times G\rightarrow \pi^{-1}(U);\,\, (x,g)\mapsto g\cdot s(x)$. Global gauge transformatios are diffeomorphisms $\tau\in \mathrm{Diff}(P)$ such that $\tau(g\cdot p)=g\cdot \tau(p)$. Locally, by equivariance, $\tau$ maps orbits to orbits, and so takes one trivializing section $s$ to another $s'$; and if $\tau(s)=\tau'(s)=s'$, then $\tau(p)=\tau'(p)$ for $p\in \pi^{-1}(U)$. Therefore, locally, gauge transformations are \textit{uniquely encoded} by the map from $s$ to $s'$. It is easy to characterize these maps: since $s$ provides a local trivialization, we can write, for all $x\in U$, $ s'(x)=g(x)\cdot s(x)$, for some $g\in C^\infty(U; G)$. Therefore, locally, gauge transformations are of the form $g\in C^\infty(U; G)$. This local construction will be mirrored for vector bundles in footnote \ref{ftnt:vector_bundle}.  \label{ftnt:PFB} } So we can write local (infinitesimal) gauge transformations as $\G=C^\infty(U,G)$ (or  $\fG= C^\infty(U,\mathfrak{g})$, respectively): the group of gauge transformations has a group operation that is just that of $G$, pointwise on $M$, i.e. for $g,g'\in \G$, $gg'(x)=g(x)g'(x)$. The group acts as the adjoint on the algebra, $\mathrm{Ad}:G\rightarrow GL({\mathfrak{g}})$, with $(g, \xi)\mapsto\mathrm{Ad}_g \xi:=g\xi g^{-1}\in \mathfrak{g}$,    and so the action of the algebra on itself is just the Lie algebra commutator, $\mathfrak{g}\times \mathfrak{g}\rightarrow \mathfrak{g}$ with $(\xi, \xi')\mapsto [\xi, \xi']\in \mathfrak{g}$.

 Let $\varphi$ be a map that takes each point of $U$ into $V$,  a vector space of dimension $n$ which is $\varphi$'s value space. We will use $i$ to indicate  components in $V$, e.g. $\varphi_i$. We will take $\varphi$ to represent matter fields. Locally, that is, on appropriate trivializing patches $U\subset M$,  $\varphi$ is a $V$-valued smooth function on $U$, the class of which is $C^\infty(U; V)$.\footnote{Here the mathematically precise description of $\varphi$ is as a section of a vector bundle $\pi:E\rightarrow M$, which one can think of as having an internal vector space isomorphic to $V$ attached to each point of $M$, that is, with $\pi^{-1}(x)\simeq V$ for $x\in M$, but without a canonical decomposition $E\simeq M\times V$. Locally, i.e. for subsets $U\subset M$ we have an isomorphism $\pi^{-1}(U)\simeq U\times V$, but the isomorphism is not canonical: writing a field that is locally valued in $V$ requires a choice of trivialisation, as in footnote \ref{ftnt:PFB} above (the standard example here is vector fields on $M$: we can only write them as maps $U\rightarrow  \RR^n$ if we fix a basis for the tangent spaces). Thus $\varphi$ is a \textit{section} of $E$, i.e. $\varphi:M\rightarrow E$ such that $\pi\circ\varphi =\mathrm{Id}_M$, but locally, given a trivialization we can write it as function $\varphi:U\rightarrow V$. The trivialization of the vector bundle can be `soldered onto' the trivialization of the principal bundle, mentioned in footnote \ref{ftnt:PFB}, by thinking of $E$ as an `associated vector bundle'. The simplest way to think of this association is to see the elements $p\in P$ of the bundle as linear frames for $V$ (at $x\in \pi(p)$). This soldering guarantees that a gauge transformation on $P$ will also act on the matter fields in the appropriate manner, e.g. \eqref{eq:gauge_trans}.\label{ftnt:vector_bundle}} 

 To represent the forces that are sourced by $\varphi$, we take the collection of vector-valued one forms $A_\mu^a$, which take a vector of $M$ at a point of $U$ to $\mathfrak{g}$, with $\mu$ representing the spacetime components of the vector and  $a$ indicating the components in $\mathfrak{g}$.\footnote{What is a `force' and what is `matter' will be further distinguished by their transformation properties under a gauge transformation, in \eqref{eq:gauge_trans}. Matter transforms linearly, whereas forces acquire derivatives of the generator as inhomogeneous terms.} 
 These fields are associated with a dynamics by postulating a preferred real-valued action functional $S(\varphi_i, A_\mu^a)$, whose extremal values are postulated to provide the equations of motion.

We also  assume $G$ has some action (a representation) on $V$, the vector space of local field-values of the matter fields $\varphi$, defining this action pointwise as $g\cdot \varphi(x)=g(x)\cdot\varphi(x)\in V$. Let $t^{ij}_a$  be the $n$-dimensional Hermitean matrix representation on $V$ of $\mathfrak{g}$, i.e. $t:\mathfrak{g}\rightarrow GL(V)$,  where  the $a$ are indices of the Lie algebra space, in the domain of the map, and $i,j$ denote the matrix indices in the image of the map, acting linearly on $V$. Then we take the (malleable) gauge transformations, infinitesimally parametrized by $\epsilon\in \fG$, to act on our fundamental variables as:
\be\label{eq:gauge_trans}
\begin{cases}
\delta_\epsilon \varphi_i=\epsilon^a t_a^{ij}\varphi_j=(\epsilon t\varphi)_i \\
\delta_\epsilon A^a_\mu=\D_\mu\epsilon^a=\pp_\mu \epsilon^a+[\epsilon, A_\mu]^a
\end{cases}.
\ee
where the square brackets are the Lie algebra commutators. Here the $\pp_\mu  \epsilon^a$ on the right hand side of the second line echoes the $\nabla\phi$ of equation \eqref{eq:3vel-ga}, and the second term in that same side allows for a non-Abelian transformation group. These transformation rules are not as general as they could be, but neither are they arbitrary: they are the first-order terms of the Lie algebra action on the respective vector spaces --- in particular, `first-order' in the derivatives of $\epsilon$ and in powers of $A$ and $\varphi$ --- and in this sense provide an appropriate approximation of any malleable gauge transformation. We here focus on this special case, equation \ref{eq:gauge_trans}, only to simplify the presentation of our argument.


Our aim now is to constrain how the matter fields $\varphi$ couple to force fields. Let $\mathcal{L}(\varphi, \pp \varphi, A, \pp A)$ be the Lagrangian defining our action $S(\varphi,A)$, which we assume for simplicity does not depend on higher-order derivatives.\footnote{This can be justified by appeal to Ostrogradsky's theorem; see \citet{swanson2019ostrogradsky} for a philosophical discussion.} Variation along the directions of the gauge transformations above yields (with summation convention on all indices):
\begin{align}\label{eq:vanishing}
  \begin{split}
\left( \frac{\delta \mathcal{L}}{\delta \varphi_i}(t^a\varphi)_i+ \frac{\delta \mathcal{L}}{\delta \pp_\mu\varphi_i}(t^a\pp_\mu\varphi)_i+\big[ \frac{\delta \mathcal{L}}{\delta A_\nu}, A_\nu]^a+\big[\frac{\delta \mathcal{L}}{\delta\pp_\nu A_\mu}, \pp_\mu A_\nu\big]^a\right)\epsilon_a+\\
\left( \frac{\delta \mathcal{L}}{\delta \pp_\mu\varphi_i}(t^a\varphi)_i+ \frac{\delta \mathcal{L}}{\delta A_\mu^a}+\big[ \frac{\delta \mathcal{L}}{\delta \pp_\nu A_\mu}, A_\nu]^a\right)\pp_\mu\epsilon_a+\\
\frac{\delta \mathcal{L}}{\delta \pp_\nu A^a_\mu}\pp_\mu \pp_\nu \epsilon^a=0\\
\end{split}
\end{align}
Since the derivatives of $\epsilon$ are functionally independent, this equation implies that \textit{each line must vanish separately}: the first line is a consequence of rigid symmetries, and the remaining two are of malleable ones. These are the fundamental constraints on the dynamics that we propose to analyse, and the task of the remainder of this paper will be to unpack them.

The requirement that each of these lines vanishes provides a strong constraint on the form of the Lagrangian, and hence on the dynamics. This, we claim, provides the core of the Noether gauge argument. To extract interesting physical information from this constraint, there are four sectors to compare, arising from the use of either rigid or malleable symmetries, and either $A$-independent or $A$-dependent Lagrangians. We treat each sector in turn. 

The results will be: a theory with rigid symmetries can be dynamically non-trivial and complete---i.e. it will not require further constraints---when $A$ does not figure in the Lagrangian. With malleable symmetries and no $A$-dependence, the constraints demand that the dynamics be  trivial, i.e. no kinetic term for the matter field can appear in the Lagrangian. When forces have their own dynamics, that is, when  the Lagrangian is $A$-dependent, a theory with rigid symmetries may be incomplete, and require further constraints to render the dynamics of $A$ compatible with charge conservation; an example will be given. It is only in the last case, where we have malleable symmetries and $A$-dependence, that the equations of motion coupling forces to charges is automatically consistent with the conservation of charges (and so no further constraints are required). Thus we will see the power of malleable symmetries and $A$-dependence together to secure an interacting dynamics that conserves charge. And this will be our Noether gauge argument. 

\subsection{$A$-independent, rigid symmetries}\label{4.3}

First, suppose we are as in the first step of the textbook gauge argument: there is no $A$ in sight, and the symmetry is rigid, so that  $\pp_\mu  \epsilon^a = 0 = \pp_\mu \pp_\nu  \epsilon^a$. Then the vanishing of the first line of Equation \eqref{eq:vanishing} reduces to
\begin{equation}\label{eq:first-line}
\frac{\delta \mathcal{L}}{\delta \varphi_i }(t^a\varphi)_i+ \frac{\delta \mathcal{L}}{\delta \pp_\mu\varphi_i}(t^a\pp_\mu\varphi)_i=0.
\end{equation}
But by the Euler-Lagrange equations $\mathsf{E}(\varphi)_i\approx 0$, where $\mathsf{E}(\varphi)_i= \frac{\delta \mathcal{L}}{\delta \varphi_i}-\pp_\mu  \frac{\delta \mathcal{L}}{\delta \pp_\mu\varphi_i}$, we have
\begin{equation}
  \frac{\delta \mathcal{L}}{\delta \varphi_i}\approx\pp_\mu  \frac{\delta \mathcal{L}}{\delta \pp_\mu\varphi_i},
\end{equation}
where we again are using `$\approx$' to denote `on-shell' equality. Applying this to Equation \eqref{eq:first-line} we find that
\be\label{eq:A_rig}
  \pp_\mu (\frac{\delta \mathcal{L}}{\delta \pp_\mu\varphi_i}(t^a\varphi)_i)=\pp^\mu J^a_\mu(\varphi)\approx 0
\ee
where we have defined the matter current as
\be\label{eq:J_phi}
  J^a_\mu(\varphi):=\frac{\delta \mathcal{L}}{\delta \pp_\mu\varphi_i}(t^a\varphi)_i.
\ee
In summary, we have derived what is guaranteed by Noether's first theorem, that the current $J^a_\mu(\varphi)$ is conserved on-shell. Or, turning this around: symmetry requires the Lagrangian to be restricted so that $J^a_\mu(\varphi)$ defined in Equation \eqref{eq:J_phi} is divergenceless. Having constrained the space of theories in this manner, there are no more equations to satisfy: conservation of charge is consistent with the dynamics and no further constraints need to be imposed. 

\subsection{$A$-independent, malleable symmetries}

In the next case, suppose that we allow ---in addition to Section \ref{4.3}'s equations --- the ones arising from a $\pp\epsilon\neq 0$, while still not allowing for an $A$ in the theory. We get, in addition to equations \eqref{eq:J_phi} and \eqref{eq:A_rig}, from the vanishing of the second line of Equation \eqref{eq:vanishing}:
\be\label{eq:noA_mal}
\frac{\delta \mathcal{L}}{\delta \pp_\mu\varphi_i}(t^a\varphi)_i=J^a_\mu(\varphi)=0. 
\ee
So here the  conserved currents are forced to vanish. Clearly this condition is guaranteed for all field values if $\frac{\delta \mathcal{L}}{\delta \pp_\mu\varphi_i}=0$, which requires a  vanishing kinetic term. A careful analysis of more general cases reveals this is the only generic solution.\footnote{For instance, assume $\frac{\delta \mathcal{L}}{\delta \pp_\mu\varphi_i}$ depends only on $\pp_\mu\varphi_i$, then since $t^a_{ij}\varphi^i$ can take any value, we must have $\frac{\delta \mathcal{L}}{\delta \pp_\mu\varphi_i}=0$. Now, suppose $\frac{\delta \mathcal{L}}{\delta \pp_\mu\varphi_i}$ depends on $\varphi_i$ as well. Since $\varphi_i$ has no spacetime indices to match the $\mu$ of the gradient $\pp_\mu\varphi_i$,   to make a Lagrangian scalar, we  would need the $\varphi_i$ contribution to this term to itself be a scalar, call it $F(\varphi)$. So for example: $\frac{\delta \mathcal{L}}{\delta \pp_\mu\varphi_i}= \pp_\mu\varphi_i(\varphi_j\varphi^j)$, or more generally $\frac{\delta \mathcal{L}}{\delta \pp_\mu\varphi_i}=F'(\pp \varphi)_{i\mu}F(\varphi)$ (where we raise indices with an inner product of $V$); and as in the example $F(\varphi)=\varphi_j\varphi^j=0$ iff $\varphi=0$. But then the same argument as before suffices, since we can still allow $t^a_{ij}\varphi^i$ to take any value in $V$ (for an appropriate,  non-zero value of the scalar formed just from $\varphi$, e.g. the contraction $\varphi_j\varphi^j$).  Or, in other words, for $\varphi\neq 0$,  $\frac{\delta \mathcal{L}}{\delta \pp_\mu\varphi_i}(t^a\varphi)_i=0$ iff $F^{-1}(\varphi)\frac{\delta \mathcal{L}}{\delta \pp_\mu\varphi_i}(t^a\varphi)_i=0$ where $F^{-1}(\varphi)\frac{\delta \mathcal{L}}{\delta \pp_\mu\varphi_i}$ depends only on $\pp_\mu\varphi$; and thus we are back to the first, simple  case.} 

This analysis pinpoints the obstacle appearing in  the textbook gauge argument that we rehearsed in Section \ref{subs:beware}. When the matter field Lagrangian has a non-trivial kinetic term, malleable transformations cannot be variational symmetries. That is: if we impose malleable symmetries without introducing a gauge potential, we cannot consistently also allow a term in the Lagrangian including  $\pp_\mu\varphi_i$. It is to allow such terms and still retain the malleable symmetries that the next two Sections will introduce the gauge potential.

\subsection{$A$-dependent, rigid symmetries}

We first proceed precisely as in the first case,  introducing the $A$ field, but still \textit{keeping the symmetries rigid}. 
Using the equations of motion for $A$ as well as those of $\varphi$, i.e. $\mathsf{E}[A]=0$ as well as  $\mathsf{E}[\varphi]=0$, we get,   in direct analogy to \eqref{eq:A_rig}, a conserved current that is a sum of two currents:\footnote{To be explicit, the $A$-dependent terms that appear in the first line of \eqref{eq:vanishing} are $\big[ \frac{\delta \mathcal{L}}{\delta A_\nu}, A_\nu]^a+\big[\frac{\delta \mathcal{L}}{\delta\pp_\nu A_\mu}, \pp_\mu A_\nu\big]^a\approx \big[\pp_\mu \frac{\delta \mathcal{L}}{\delta\pp_\nu A_\mu}, A_\nu]^a+\big[\frac{\delta \mathcal{L}}{\delta\pp_\nu A_\mu}, \pp_\mu A_\nu\big]^a=\pp_\mu\big[ \frac{\delta \mathcal{L}}{\delta \pp_\nu A_\mu}, A_\nu]^a$. } 
\be\label{eq:conserv} \pp_\mu \left(\frac{\delta \mathcal{L}}{\delta \pp_\mu\varphi_i}(t^a\varphi)_i+\big[ \frac{\delta \mathcal{L}}{\delta \pp_\nu A_\mu}, A_\nu]^a\right)=\pp^\mu (J^a_\mu(\varphi)+\tilde J^a_\mu(A))\approx 0
\ee
and nothing more; there are  no further conditions that the terms of the Lagrangian need to obey.  (Here, the definition of $\tilde J^a_\mu(A))$ is given by \eqref{eq:conserv}.)

So, unlike the previous case, which admitted only a trivial kinetic term for the matter field $\varphi$, this sector will admit many possible dynamics. The problem here is of a different nature: the theories are not sufficiently constrained; the equations of motion do not automatically guarantee conservation of charges.

Let us look at an example of how things can go wrong in this intermediate  sector containing forces but only rigid symmetries, for the simple, Abelian theory. In the Abelian theory, $\tilde J(A)\equiv 0$, since quantities trivially commute. Thus Equation \eqref{eq:conserv} only contains the standard conservation of the matter charges and the symmetries are silent about the relationship between this charge and the dynamics of the forces. 

Consider a kinetic term of the form $\pp_{(\mu}A_{\nu)}\pp^{(\mu}A^{\nu)}$ where round brackets denote symmetrization. So this differs from the standard Maxwell theory kinetic term for the gauge potential: namely, $F_{\mu\nu}F^{\mu\nu}:=\pp_{[\mu}A_{\nu]}\pp^{[\mu}A^{\nu]}$ where square brackets denote \textit{anti}-symmetrization. But the symmetrized version is nonetheless gauge-invariant (under {\em rigid} transformations). Now, the Euler-Lagrange equations for this theory differ only very slightly from the Maxwell-Klein-Gordon equations. The equations of motion for $A$ yield:
\be\label{eq:wrongL}
\pp^\mu (\pp_{(\mu}A_{\nu)})=J_\nu
\ee
 in contrast with the usual $\pp^\mu (\pp_{[\mu}A_{\nu]})=J_\nu$. 
But clearly, unlike the usual case, the divergence of the left hand side does \textit{not} automatically vanish:
\be\label{eq:counter_exA}
\pp^\nu\pp^\mu (\pp_{(\mu}A_{\nu)})=\pp^\mu\pp_\mu\pp^\nu A_\nu=\square\pp^\nu A_\nu\not\equiv0.
\ee
At this point, we would have to go back to the drawing board and introduce more constraints on the theory: this theory does not couple forces to charges in a manner that guarantees charge conservation.

Thus we glimpse our overall thesis: only by introducing malleable gauge symmetries do we restrict interactions between forces and their sources so that they are consistent with the conservation of the matter current. 

Of course, in this example it is easy to see what is the smoking gun: the kinetic term $\pp_{(\mu}A_{\nu)}\pp^{(\mu}A^{\nu)}$ is \textit{not} invariant under malleable transformations. According to the next Section---our fourth sector---requiring this stronger form of invariance will restrict us to the space of consistent interactions.  No tweaking required.

\subsection{$A$-dependent, malleable symmetries}
In this fourth sector, we again obtain \eqref{eq:conserv}, from the vanishing of the first line of \eqref{eq:vanishing}, since nothing changes at that level. But, from the vanishing of the second line in Equation \eqref{eq:vanishing}, we have:
\be
 -\frac{\delta \mathcal{L}}{\delta A_\mu^a}=\frac{\delta \mathcal{L}}{\delta \pp_\mu\varphi_i}(t^a\varphi)_i+\big[ \frac{\delta \mathcal{L}}{\delta \pp_\nu A_\mu}, A_\nu]^a=J^a_\mu(\varphi)+\tilde J^a_\mu(A).
\ee
Once again using the Euler-Lagrange equations for $A$, to substitute the left-hand side, we find that
\be \mathsf{E}(A)^a_\mu=\frac{\delta \mathcal{L}}{\delta A_\mu^a}-\pp_\nu\frac{\delta \mathcal{L}}{\delta \pp_\nu A^a_\mu}\approx 0.
\ee
Defining $\frac{\delta \mathcal{L}}{\delta \pp_\nu A^a_\mu}=:k_{\mu\nu}^a$, we now obtain: 
\be\label{eq:final_conserv} J^a_\mu(\varphi)+\tilde J^a_\mu(A) = -\pp^\mu k_{\mu\nu}^a+  \mathsf{E}(A)^a_\mu\approx -\pp^\mu k_{\mu\nu}^a
\ee
This equation links both the matter and force currents to the dynamics of the force field. 

We already know from the vanishing in the first line of Equation \eqref{eq:vanishing} that the sum of the currents is divergence-free on shell (cf.\ Equation \ref{eq:conserv}). Thus, taking the divergence on the left hand side of \eqref{eq:final_conserv}, we must have $\pp^\nu \pp^\mu k_{\mu\nu}^a= 0$. Since two derivatives of a scalar field are necessarily symmetric, all we need in order to satisfy conservation is that:
\be
k_{\mu\nu}^a=-k_{\nu\mu}^a\quad \text{or}\quad k_{\mu\nu}^a=k_{[\nu\mu]}^a,
\ee
 which is just what we have from the vanishing of the \emph{third} line of Equation \eqref{eq:vanishing}. Thus, the result of including malleable symmetries, in this simple case, restricts us to consider Lagrangians in which the derivatives of $A^a_\mu$  only enter in anti-symmetrized form: $\pp_{[\mu}A^a_{\nu]}$. This restriction excludes the previous example of equation \eqref{eq:wrongL}. 

More generally, if we try to find a Lagrangian that includes  force fields without obeying the relations obtained from the malleable symmetries, the equations of motion of the force fields and those relating force fields and matter may require further constraints to be compatible with charge conservation, as we saw in the counter-example in the previous section. This is the power of local gauge symmetries: they link charge conservation --- taken as empirical fact or on a priori grounds --- with the form of the Lagrangian for the force fields. 


\subsection{Masslessness: An invitation}
There is yet more information that can be gleaned from the Noether gauge argument, which is contained in equation \eqref{eq:final_conserv}: upon integration, it yields a boundary term and a volume integral. That is, it gives  a relation between a quantity at a far-away boundary --- related to the flux of the force field components $A^a_\mu$ --- and the matter content inside this region. In the Abelian case, for the $0$th component of the equations of motion, this just gives the standard Gauss law. But more generally, being detectable at arbitrarily long distances makes the `forces' associated to $k_{\mu\nu}$ long-range.\footnote{This is a classical treatment. Quantum mechanically, non-Abelian theories would suffer from confinement, which lies outside the scope of this discussion.} 

It is common to conclude\footnote{For example, in the context of quantum electrodynamics, compare \citep[p. 343]{weinbergQFTv1}.} on this basis that gauge invariance \emph{forbids the presence of a mass term} for $A_\mu$. However, like the textbook gauge argument, the general form of this argument for masslessness is often heuristic in character. In particular, it assumes that each term in the Lagrangian is independently gauge-invariant. Then it is true that, on its face, a term like  $m^2 A^a_\mu A_a^\mu$ is not invariant under malleable symmetries.\footnote{Note that the Proca action does include a mass term for the photon field, i.e. the gauge potential $A$, but it is only gauge-invariant with $m=0$, in which case it just reduces to the standard Maxwell equations. For $m\neq 0$, one must have, in relation to the Maxwell equations, gauge-breaking, or `gauge-fixing', conditions. For a discussion, see \citet[\S 3-2-3]{ItzyksonZuber}.} 

To show in full generality that masslessness is required would go beyond the scope of this paper: we leave it as an exercise to the ambitious reader to explore! However, we can still improve on the standard heuristic argument without much effort in a special case that includes electromagnetism, by enforcing Equation \eqref{eq:vanishing} off-shell, for all models $(\varphi, A)$, and requiring that any mass term $m$ be field-independent and that the only coupling between the matter field and $A$ be just $A_\mu^a J^\mu_a$. 

A term in the Lagrangian of the form $m^2 A^a_\mu A_a^\mu$ would not leave any trace in the first line of Equation \eqref{eq:vanishing}.   But from the second line, again assuming no on-shell constraint between values of the matter and force field, from the second term, we obtain $m^2A_\mu^a$, which can only cancel with something dropping out of $\big[ \frac{\delta \mathcal{L}}{\delta \pp_\nu A_\mu}, A_\nu]^a$. Call this term $\kappa_{\mu\nu}(A, \pp A)$, which is such that $[\kappa_{\mu\nu}, A^\mu]^a\propto A^a_\nu$ for all $A$. Take, in this basis, $A^a_\nu$ to be constant, e.g. $\delta^a_{1}\delta_\nu^{y}$, so that $A=\tilde A:=\sigma_1\otimes dy$, where $y$ is a spacetime coordinate function and $\sigma_1$ is one element of the Lie-algebra basis.  This implies that the partial derivatives inside $\kappa_{\mu\nu}(\pp A, A)$ vanish, that is, $\kappa_{\mu\nu}(\pp \tilde A, \tilde A)=\kappa_{\mu\nu}(\tilde A)$, and therefore that it is a polynomial of $A$ with no derivatives, i.e. it is a polynomial that contains a single element of the  Lie-algebra basis, $\sigma_1$, and the spacetime 1-form basis, $dy$. But this means that the commutator $[\kappa_{\mu\nu}, \tilde A^\mu]^a$ vanishes, and therefore cannot be proportional to $\tilde A$. As a result, in order to maintain off-shell invariance under malleable transformations, a mass term for $A$ cannot be included in the Lagrangian.



\section{Conclusion}\label{sec:conclusion}

We have given a detailed defence of the use of gauge symmetries for theory-building, in the spirit of the textbook gauge argument.

We first showed how the textbook gauge argument can be tightened into a rigorous theorem in quantum mechanics, which begins with the assumption that malleable transformations `appropriately' transform velocity, and concludes that the dynamics must be given by the minimally coupled Hamiltonian (Section \ref{sec:ga-as-a-thm} and the \hyperref[sec:appendix]{Appendix}).

We then went on to defend a much more general `Noether gauge argument' in classical field theory. In particular, gauge symmetries of various kinds were fed into the powerful theorems of Emmy Noether, in order to produce precise constraints on the possible dynamics. The result is more than a simple argument that gauge symmetry is useful for theory construction: gauge symmetry constrains how one can consistently combine charges with the fields they interact with. Noether's first theorem of course implies charge conservation; but the second theorem then implies relations between the theory's equations of motion, which amounts to a coupling constraint. In other words, converting a rigid symmetry into a malleable one enforces the compatibility between charge conservation and the dynamics of the corresponding fields: a result which has not yet been stressed by the philosophical literature.

Of course, one may still feel that gauge redundancies are like  Wittgenstein's ladder: once our theories have been successfully constructed, why not throw gauge symmetries away, and move down to a description in which such redundancies have been eliminated? We would reply: on the contrary, the ladder remains invaluable in the interpretation of gauge theories. In addition to the several reasons for gauge identified in the Introduction to this paper, we find that gauge symmetries provide an explanatory reason --- a reason drawing heavily on Noether's two theorems --- for the way charges couples to fields.
 
\subsection*{Acknowledgements}
We would like to thank Adam Caulton, Neil Dewar, Caspar Jacobs, Ruward Mulder, James Read, and an anonymous referee for valuable comments.

\section*{Appendix}\label{sec:appendix}

\GAhamiltonian*

\begin{proof}
We first note that if $W_\phi:=e^{i\phi}:=e^{i\phi(Q)}$, then
\begin{align}\label{eq:qp-rules}
  W_\phi QW^*_\phi = Q, && W_\phi PW^*_\phi = P - \nabla\phi.
\end{align}
The second equation follows from the observation that $PW_{\phi}^*\psi = -i\nabla \left(e^{-i\phi}\psi\right) = -i\left(e^{-i\phi}\nabla\psi - i(\nabla \phi)e^{-i\phi}\psi\right) = e^{-i\phi}(-i\nabla - \nabla \phi)\psi = W_\phi^*(P - \nabla\phi)\psi$, for all $\psi\in L^2(\RR^3)$. From the second equation, and our assumption that $W_\phi\dot{Q}W^*_\phi = \dot{Q}-\nabla\phi$, it follows that $[W_\phi,m\dot{Q}-P]=0$ for any $m\neq0$. Since this holds in particular when $\phi(x)=ax$ for all $a\in\RR$, it follows that $m\dot{Q}-P$ commutes with $Q$ \citep[Proposition 5.9.3]{BlankExnerHavlicek}. But the position operator $Q=x$ in the Schr\"odinger representation has a simple spectrum \citep[Example 5.8.2]{BlankExnerHavlicek}, which is to say that its commutant $\{Q\}'$ is equal to its bicommutant $\{Q\}''$. (This is the continuous spectrum analogue for $Q$ of all of a quantity's eigenvalues being non-degenerate.) By von Neumann's bicommutant theorem \citep[Theorem 5.5.6]{BlankExnerHavlicek} this implies that there is a function $A_r$ (for each $r=1,2,3$) of position alone such that,
\begin{equation}\label{eq:A}
  A_r(x) = P_r - m\dot{Q}_r.
\end{equation}

Using the fact that $A$ depends only on position and thus commutes with $Q$, we can now observe by direct calculation that for each $r = 1,2,3$: 
\begin{align} 
  \begin{split}
    [(P - A)^2,Q] & = \underbrace{[P^2, Q]}_{-2iP} - [PA,Q] - [AP,Q] + \underbrace{[A^2,Q]}_{0}\\
       & = -2iP - (P\underbrace{AQ}_{=QA} - QPA + APQ - \underbrace{QA}_{=AQ}P)\\
       & = -2iP - (\underbrace{[P,Q]}_{-i}A) + A\underbrace{[P,Q]}_{-i})\\
       & = -2i(P - A),
   \end{split}
\end{align}
and thus $i[\tfrac{1}{2m}(P_r-A_r)^2,Q_r] = \tfrac{1}{m}(P_r-A_r)$. But also, $i[H,Q_r] = \dot{Q}_r = \tfrac{1}{m}(P_r - A_r)$, where the second equality applies Equation \eqref{eq:A}, so $[H- \tfrac{1}{2m}(P_r-A_r)^2, Q_r] = 0$. By the simple spectrum property it again follows that there is a function $V$ of position alone such that $H - \sum_r\tfrac{1}{2m}(P_r - A_r)^2 = V$, which establishes Equation \eqref{eq:ham-form}. 

To confirm Equation \eqref{eq:gauge}, let $\psi(t) := \tilde{U}_t\psi := W_{\phi}U_t =: e^{-itH_\phi}\psi$ for each $\psi\in L^2(\RR)$. Then on the one hand we have that, 
\begin{equation}\label{eq:diff1}
  \tfrac{d}{dt}\psi(t) = -iH_\phi e^{-itH_\phi}\psi = -iH_\phi\psi(t).
\end{equation}
On the other hand, using our definition $\tilde{U}_t\psi := W_\phi U_t\psi$, we have that,
\begin{align}\label{eq:diff2}
  \begin{split}
    \tfrac{d}{dt}\psi(t) & = \left( W_\phi(\tfrac{d}{dt}U_t) + (\tfrac{d}{dt}W_\phi)U_t \right)\psi = -i(W_\phi HU_t - \tfrac{d\phi}{dt}W_\phi U_t) \psi\\
    & = -i\big(W_\phi H W^*_\phi - \tfrac{d\phi}{dt}\big)(W_\phi U_t)\psi\\
    & = -i\left(W_\phi H W^*_\phi - \tfrac{d\phi}{dt}\right)\psi(t).\\      
  \end{split}
\end{align}
Equations \eqref{eq:diff1} and \eqref{eq:diff2} together imply $H_\phi = W_\phi H W^*_\phi - \tfrac{d}{dt}\phi$. Moreover, since $W_\phi$ is a function of $Q$ and so commutes with functions of $Q$, it follows that $W_\phi (P-A)W^*_\phi = P - \nabla\phi - A = P - (A + \nabla\phi)$. Therefore, using Equation \eqref{eq:ham-form} and rearranging, we find that
\begin{equation}
  H_\phi = W_\phi H W^*_\phi - \tfrac{d\phi}{dt} = \sum_{r=1}^3\tfrac{1}{2m}\left(P_r - (A_r+\nabla\phi)\right)^2 + (V-\tfrac{d\phi}{dt}) \, ;
\end{equation}
which proves Equation \eqref{eq:gauge} and  claim (ii). 

The final claim (iii) is now verified by confirming that if $i\tfrac{d}{dt}\psi = H\psi$ and $\psi_\phi:=W_\phi\psi$, then we have that $i\tfrac{d}{dt}\psi_\phi \equiv i\tfrac{d}{dt}(W_\phi\psi) = W_\phi(-\tfrac{d\phi}{dt}\psi + i\tfrac{d}{dt}\psi) = W_\phi(H - \tfrac{d\phi}{dt})\psi = (W_\phi H W_\phi^* - \tfrac{d\phi}{dt})W_\phi\psi = H_\phi\psi_\phi$.
\end{proof}
 
\setstretch{1.0}

\end{document}